\documentclass[11pt, a4paper]{article}

\usepackage{graphicx}
\usepackage{amsfonts,amsmath,amssymb}
\usepackage{hyperref}
\usepackage{amsthm}
\usepackage{multirow}
\usepackage{xspace}
\usepackage{url}
\usepackage{comment}
\usepackage[algo2e,vlined,linesnumbered,tworuled]{algorithm2e}

\usepackage{tikz}
\usepackage{pgf}
\usepackage{pgfplots}

\SetAlFnt{\footnotesize}
\SetAlCapFnt{\footnotesize}
\SetAlCapNameFnt{\footnotesize}
\SetVlineSkip{1pt}
\SetAlCapHSkip{0pt}

\newcommand{\MICM}{{\textsc{BIMA}}\xspace}
\newcommand{\CICM}{{\textsc{CostIMA}}\xspace}
\newcommand{\Greedy}{{\textsc{Greedy}}\xspace}

\newtheorem{theorem}{Theorem}
\newtheorem{lemma}[theorem]{Lemma}
\newtheorem{proposition}[theorem]{Proposition}
\newtheorem{corollary}[theorem]{Corollary}

\title{Selecting nodes and buying links to maximize the information diffusion in a network \\\normalsize{(Extended Version)}}
\author{Gianlorenzo D'Angelo,$^{1}$ Lorenzo Severini,$^{2}$ Yllka Velaj$^{1,3}$\\
\normalsize{$^{1}$Gran Sasso Science Institute (GSSI).}\\
\normalsize{$^{2}$ISI Foundation.}\\
\normalsize{$^{3}$University of Chieti-Pescara.}\\
}
\date{}

\begin{document}
\maketitle
\begin{abstract} 
The Independent Cascade Model (ICM) is a widely studied model that aims to capture the dynamics of the information diffusion in social networks and in general complex networks. In this model, we can distinguish between active nodes which spread the information and inactive ones. The process starts from a set of initially active nodes called seeds. Recursively, currently active nodes can activate their neighbours according to a probability distribution on the set of edges. After a certain number of these recursive cycles, a large number of nodes might become active. The process terminates when no further node gets activated.

Starting from the work of Domingos and Richardson~\cite{DR01,RD02}, several studies have been conducted with the aim of shaping a given diffusion process so as to maximize the number of activated nodes at the end of the process. One of the most studied problems has been formalized by Kempe et al. and consists in finding a set of initial seeds that maximizes the expected number of active nodes under a budget constraint~\cite{KKT03}.
In this paper we study a generalization of the problem of Kempe et al. in which we are allowed to spend part of the budget to create new edges incident to the seeds. That is, the budget can be spent to buy seeds or edges according to a cost function. The problem does not admin a PTAS, unless $P=NP$. We propose two approximation algorithms: the former one gives an approximation ratio that depends on the edge costs and increases when these costs are high; the latter algorithm gives a constant approximation guarantee which is greater than that of the first algorithm when the edge costs can be small.
\end{abstract}

\section{Introduction}
When a new idea or innovation arises in a network of individuals, it can either quickly propagate to a large part of the network and be adopted by many individuals or immediately expire. Understanding the dynamics that regulate these behaviours has been one of the main goals in the field of complex network analysis and has been studied under the name of \emph{influence spreading} or \emph{information diffusion} analysis problem~\cite{DR01,KKT03}.
The motivating application span several fields: from marketing with the aim of evaluating the success of a new product or maximizing its adoption~\cite{B69,CWW10,GLM01,MMB90,RD02},
to epidemiology in order to limit the diffusion of a virus or disease~\cite{M07,N02}, the study of adoption of innovations~\cite{CKM66,R95,V95}, the analysis of social networks to find influential users and to study how information flows through the network~\cite{BHMW11}, and the analysis of cascading failures in power networks~\cite{ARLV01}.
 
Different models of information diffusion have been introduced in the literature. Two widely studied models are the \emph{Linear Threshold Model} (LTM)~\cite{G78,KKT15,S06} and the \emph{Independent Cascade Model} (ICM)~\cite{GLM01,GLM01a,KKT03,KKT15}.
In both models, we can distinguish between \emph{active}, or \emph{affected}, nodes which spread the information and inactive ones. At the beginning of the process a small percentage of nodes of the graph is set to active in order to let the information diffusion process start. These nodes are called \emph{seeds}. Recursively, currently affected nodes can infect their neighbours with some probability.  After a certain number of these cascading cycles, a large number of nodes might becomes affected in the network.
In LTM the idea is that a node becomes active as more of its neighbours become active. Formally, each node $u$ has a threshold $t_u$ chosen uniformly at random in the interval $[0,1]$. The threshold represents the weighted fraction of neighbours of $u$ that must become active in order for $u$ to become active. During the process, a node $u$ becomes active if the total weight of its active neighbours is greater than $t_u$.
In ICM, instead, an active node $u$ tries to influence its inactive neighbours but the success of node $u$ in activating the node $v$ only depends on the propagation probability of the edge from $u$ to $v$ (each edge has its own value). Regardless of its success, the same node will never get another chance to activate the same inactive neighbour.
The process terminates when no further node gets activated.

An interesting question, in the analysis of the information diffusion through a network, is how to shape a given diffusion process so as to maximize or minimize the number of activated nodes at the end of the process by taking intervention actions. Many intervention actions have been studied in the literature, the most important one has been proposed by Domingos and Richardson in the field of viral marketing and asks to find a small set of ``influential'' seeds in a network in order to activate a large part of the network~\cite{DR01,RD02}. The problem has been formalized by Kempe et al. as follows: if we are allowed to choose at most $k$ seeds, which ones should be selected so as to maximize the number of active nodes resulting from the diffusion process~\cite{KKT15}. 
This problem admits a $(1-\frac{1}{e})$-approximation algorithm and this factor cannot be improved, unless $P=NP$~\cite{KKT15}.
Besides seeds selection, other intervention actions may be used to facilitate the diffusion processes, such as inserting or deleting edges and adding or deleting nodes in the network.
Since in social networks and in other complex networks users can add edges incident to themselves, in this paper we consider the possibility to create a limited number of new edges incident to the initial seed nodes.
In detail, we study the following generalization of the problem of Kempe et al.: we are given a cost for each possible edge that can be created in a network and a budget $k$, we want find a set of seed nodes $A$ and a set of edges $S$ incident to the nodes in $A$ such that the expected number of active nodes at the end of the diffusion process is maximized and the overall cost of $A$ and $S$ does not exceed $k$, assuming that all the seeds have the same cost.

\noindent
\textbf{Related work.}
The problem of choosing initially active nodes to maximize the information diffusion in a network has been widely studied, we refer the interested reader to~\cite{KKT15} and references therein for more detail, while the budgeted version of the seed selection problem was proposed in~\cite{NG13}. In the following we focus on the problem of modifying a graph in order to maximize or minimize the spread of information through a network under LTM and ICM models.

To the best of our knowledge, under LTM, the problems that have been studied are those outlined in what follows. Khalil et al.~\cite{KDS14} consider two types of actions, adding edges to  or  deleting  edges  from  the  existing  network to minimize the information diffusion and they show that this network structure modification problem has a supermodular objective and therefore can be solved by algorithms with provable approximation guarantees. Zhang et al.~\cite{ZAVP15} consider arbitrarily specified set of nodes, and interventions that involve both edge and node removal from the set. They develop algorithms with rigorous performance guarantees and good empirical performance. Kimura et al.~\cite{Kimura08} use a greedy approach to delete edges under the LTM  but do not provide any rigorous approximation guarantees. Kuhlman et al.~\cite{Kuhlman} propose heuristic algorithms for edge removal under a simpler deterministic variant of LTM which is  not  only  hard,  but  also  has  no  approximation  guarantee.  Papagelis~\cite{P15} and Crescenzi et al.~\cite{CDSV16} study the problem of augmenting the graph in order to increase the connectivity or the centrality of a node, respectively and experimentally show that this increases the average number of eventual active nodes.

Under ICM, Wu et al.~\cite{WSZ15} consider intervention actions other than edge addition, edge deletion and seed selection, such as increasing the probability that a node infects its neighbours. They proved that optimizing the selection of these actions with a limited budget is $NP$-hard and is neither submodular nor supermodular. Sheldon et al.~\cite{Sheldon} study the problem of node addition to maximize the spread of information, and provide a counterexample showing that the objective  function  is not submodular. Kimura et al.~\cite{Kimura} propose methods for efficiently finding good approximate solutions on the basis of a greedy strategy for the edge deletion problem under the ICM, but do not provide any approximation guarantees. D'Angelo et al.~\cite{DAngeloSV16} introduce a preliminary
version of the edge addition problem where new edges are incident to a given initial set of active nodes. They focus on the case in which the initial set of seeds is a singleton and they investigate the existence of a constant approximation algorithm.

\noindent
\textbf{Our results.}
In this paper, we focus on the independent cascade model and investigate the problem of selecting a set of initial seeds and adding a small number of edges incident to the seeds, without exceeding a given budget $k$, in order to maximize the spread of information in terms of the expected number of nodes that eventually become active. The problem we analyse differs from above mentioned ones since we make the reasonable restriction that the edges to be added can only be incident to the seed nodes. To our knowledge, similar problems have never been studied for the independent cascade model. We refer to this problem as the \emph{Budgeted Influence Maximization with Augmentation problem} (\MICM).

We observe that the \MICM problem is a generalization of the problem in~\cite{KKT15} and therefore cannot be approximated within a factor greater than $1-\frac{1}{e}$, unless $P=NP$. 

We then focus on approximation algorithms. We first assume that the edge costs are all greater that a given constant $c_{min}$ and, in Section~\ref{sec:algo_min}, we propose two approximation algorithms that guarantee approximation factors of $1-\frac{1}{e^{\frac{c_{\min}}{1+c_{\min}}}}$ and $\left(1-\frac{1}{e}\right)c_{\min}$, respectively. Note that these factors increase with $c_{\min}$ and the second algorithm achieves the optimal approximation of $1-\frac{1}{e}$ when all the edge costs are equal to 1.
It is somewhat expected that when the all costs are high the approximation factor gets closer to the lower bound of $1-\frac{1}{e}$. Indeed, when the cost of buying an edge approaches the cost of buying the node at the tail of this edge, we can buy the node instead of the edge with small increase in the cost. This allows us to exploit the $1-\frac{1}{e}$ approximation algorithm  for the seed selection problem proposed in~\cite{KKT15}.
However, the challenge of our problem consists in finding a good approximation even when the cost function includes small values. 
In Section~\ref{sec:algo}, we focus on the general case and propose an algorithm that guarantees a constant approximation ratio of about $0.0878$. This algorithm outperforms the other two algorithms when the cost is allowed to be small. See Figure~\ref{fig:apxratio} for a summary of the approximation ratios of our algorithms.

\section{Preliminaries}\label{sec:definition}
A social network is represented by a weighted directed graph $G=(V,E,p,c)$, where $V$ represents the set of nodes, $E$ represents the set of relationships, $p:V\times V \rightarrow [0, 1]$ is the propagation probability of an edge, that is the probability that the information is propagated from $u$ to $v$ if $(u,v)\in E$, and $c:V\times V \rightarrow [0, 1]$ is the cost of adding an edge to $E$.

In ICM, each node can be either \emph{active} or \emph{inactive}. If a node is active (or adopter of the innovation), then it is already influenced by the information under diffusion. If a node is inactive, then it is unaware of the information or not influenced.
The process runs in discrete steps. At the beginning of the ICM process, few nodes are given the information, they are known as \emph{seed nodes}. Upon receiving the information these nodes become active. In each discrete step, an active node tries to influence one of its inactive neighbours.  The success of node $u$ in activating the node $v$ depends on the propagation probability of the edge $(u, v)$, independently of the history so far. In spite of its success, the same node will never get another chance to activate the same inactive neighbour. The process terminates when no further nodes become activated from inactive state.

We define the influence of a set $A\subseteq V$ in the graph $G$, denoted $\sigma(A)$, to be the expected number of active nodes in $G$ at the end of the process, given that $A$ is the set of seeds. Given a set $S$ of edges not in $E$, we denote by $G(S)$ the graph augmented by adding the edges in $S$ to $G$, i.e. $G(S) = (V, E \cup S)$. We denote by $\sigma(A,S)$ the influence of $A$ in $G(S)$.

In this paper we look for a set of seeds $A$ and a set of edges $S$, to be added to $G$, incident to these seeds that maximize $\sigma(A,S)$. W.l.o.g. we assume that each seed node can be selected with cost $1$, while each edge $e\in (V\times V)\setminus E$ can be selected with cost $c_e\in [0,1]$. In detail, the \MICM problem is defined as follows: given a graph $G = (V, E)$ and a budget $k$, find a set $A$ of seeds and a set $S$ of edges such that $S \subseteq (A\times V)\setminus E$, $c(A,S)\leq k$, and $\sigma(A, S) $ is maximum, where $c(A,S) = |A| + \sum_{e\in S}c_e$. 
% Without loss of generality, we assume that for some $e\in (V\times V)\setminus E$ $k>1+c_e$.  SERVE????

A \emph{live-edge graph} $X=(V, E_X)$ of $G$ is a directed graph where the set of nodes is the same set $V$ and the set of edges $E_X$ is a subset of $E$ given by an edge selection process in which each edge in $E$ belongs to $E_X$ or not according to its propagation probability. In detail, 
we can assume that for each edge $e=(u,v)$ in the graph, we flip a coin of bias $p_e$ and only the edges for which the coin indicated an activation belong to $E_X$. It is easy to show that the information diffusion process is equivalent to a reachability problem in live-edge graphs: given any seed set $A$, the distribution of active node sets after the diffusion process ends is the same as the distribution of node sets reachable from $A$ in live-edge graphs.
We denote by $\chi$ the probability space in which each sample point specifies one possible set of outcomes for all the coin flips on the edges, that is the set of all possible live-edge graphs of $G$. For a set of edges $S\subseteq (V\times V)\setminus E$, the set of all possible live-edge graphs of $G(S)$ is denoted by $\chi(S)$. Given two set of edges $S$, $T$, such that $S\subseteq T$, for each live-edge graph $X$ in $\chi(S)$ we denote by $\chi(T,X)$ the set of live-edge graphs in $\chi(T)$ that have $X$ as a subgraphs and possibly contain other edges in $T\setminus S$. In other words, a live-edge graphs in $\chi(T,X)$ has been generated with the same outcomes as $X$ on the coin flips in the edges of $E\cup S$ and it has other outcomes for edges in $T\setminus S$. The following holds: $|\chi(T,X)| = 2^{|T\setminus S|}$, for each $Y\in\chi(T,X)$ $\mathbb{P}[Y] = \mathbb{P}[Y|X]\mathbb{P}[X]$,  $\mathbb{P}[X] = \sum_{Y\in\chi(T,X)}\mathbb{P}[Y]$, and $\sum_{Y\in\chi(T,X)} \mathbb{P}[Y|X] = 1$.
For a node $a\in V$ and a live-edge graph $X$ in $\chi(S)$, let $R(a, X)$ be the set of all nodes that can be reached from $a$ in graph $X$, that is for each node $u\in R(a,X)$, there exists a path from $a$ to $u$ consisting entirely of live edges with respect to the outcome of the coin flips that generates $X$. Let $R(A, X)=\bigcup_{a \in A} R(a, X)$, then $\sigma(A, S)$ can be computed as $\sigma(A, S)=\sum_{X \in \chi(S)} \mathbb{P}[X]\cdot| R(A, X)|$.

Computing $\sigma(A)$ is $\# P$-complete~\cite{CWW10}, however it has been proven by using the Chernoff bound that it can be  approximated within an arbitrarily good factor by simulating the random process a polynomial number of times~\cite{KKT15}. Therefore, in the rest of the paper we can assume that we can compute $\sigma(A)$ (and $\sigma(A,S)$) within an arbitrary bound. This reflects to an additional factor $1+\epsilon$, for any $\epsilon>0$, to all algorithms presented in this paper. For the sake of clarity, we omit this factor from the approximation factor of our algorithms.

Given a set of edges $S$, for each graph $X\in\chi(S)$ and subset of edges $T\subseteq S$, we denote by $X^T$ the graph obtained by removing edges in $T$ from $X$.
% To avoid cumbersome notation, when $X\setminus T = \{e\}$ we denote $X^{e} = X^{\{e\}}$.
Given two feasible solutions $(A_1,S_1)$ and $(A_2,S_2)$, such that $A_2\subseteq A_1$ and $S_2\subseteq S_1$, we denote with $\delta(A_1,S_1,A_2,S_2)$ the expected number of nodes affected by $(A_1,S_1)$ and not affected by $(A_2,S_2)$, formally:
$$\delta(A_1,S_1,A_2,S_2) = \sum_{X \in \chi(S_1)} \mathbb{P}[X]\cdot\left(|R(A_1, X)| - |R(A_2, X^T)|  \right),$$ where  $T=S_1\setminus S_2$.
\begin{proposition}\label{prop:delta}
 For each $A_2\subseteq A_1\subseteq V$ and $S_2\subseteq S_1\subseteq V\times V$, such that the edges in $S_1$ and $S_2$ are outgoing $A_1$ and $A_2$, respectivey, then  $\delta(A_1,S_1,A_2,S_2) = \sigma(A_1,S_1) - \sigma(A_2,S_2)$.
\end{proposition}
\begin{proof}
\begin{align*}
\delta(A_1,S_1,A_2,S_2) &= \sum_{X \in \chi(S_1)} \mathbb{P}[X]\cdot\left(|R(A_1, X)| - |R(A_2, X^T)|  \right)\\
& = \sigma(A_1,S_1) - \sum_{X \in \chi(S_1)} \mathbb{P}[X]\cdot|R(A_2, X^T)|\\
& = \sigma(A_1,S_1) - \sum_{X \in \chi(S_2)} \mathbb{P}[X] \sum_{Y \in \chi(S_1)} \mathbb{P}[Y|X]\cdot|R(A_2, X^T)|\\
& = \sigma(A_1,S_1) - \sigma(A_2,S_2).\qedhere
\end{align*}
\end{proof}

We observe that our problem is a generalization of the influence maximization problem in~\cite{KKT15}, indeed it is enough to set $c_e=1$ for each $e\in (V\times V)\setminus E$. It follows that \MICM cannot be approximated within a factor greater than $1-\frac{1}{e}$, unless $P=NP$.

\section{Lower-bounded edge costs}\label{sec:algo_min}
In this section we consider the case in which the edge costs are at least a given value $c_{\min}$, that is for each $e\in V\times V, c_e\geq c_{\min}$. It can be easily shown that in this case selecting a set $A$ of $k$ seed nodes that maximizes $\sigma(A,\emptyset)$ guarantees an approximation factor of $c_{\min}$. Since this problem can be optimally approximated within $1-\frac{1}{e}$~\cite{KKT15}, we can obtain an overall $\left(1-\frac{1}{e}\right)c_{\min}$ approximation. In what follows we give an approximation algorithm that improves over this bound for small values of $c_{\min}$. The following analysis serves also as a warm-up for the analysis of the algorithm proposed in the next section.

% a greedy algorithm that, like those in~\cite{CDSV16,DSV16}, at each iteration selects the seed or the edge that maximizes the expected number of affected nodes is not suitable for the \MICM problem because it does not take into account the cost of the added edges. Moreover, a greedy algorithm that selects at each iteration a seed $a$ and/or and edge $(a,v)$ that maximizes the ratio between $\delta(A\cup\{a\},S\cup\{(a,v)\},A,S)$ and the marginal cost of $a$ and $(a,v)$ exhibit an unbounded approximation ratio. Indeed, ...

\begin{algorithm2e}[t]
\caption{}
\label{alg:greedy}
\SetKwInput{Proc}{Algorithm}
%\Proc{\Greedy}
\SetKwInOut{Input}{Input}
\SetKwInOut{Output}{Output}
\Input{A directed graph $G=(V,E)$ and an integer $k\in\mathbb{N}$}
\Output{A set of nodes $A$ and a of edges $S\subseteq (A\times V)\setminus E$ such that $c(A,S)\leq k$}
$A:=\emptyset$; $S:=\emptyset$;
$U:=V$;
$T:=(V\times V)\setminus E$\;
\While{$T\neq \emptyset$ or $U\neq \emptyset$}
{\label{alg:greedy:whilestart}
  $r_1 = \max_{a \in U} \{ \delta(A\cup \{a\},S, A, S) \}$\;\label{alg:greedy:r1}
  $r_2 = \max_{(a,v)\in (A\times V)\cap T} \left\{ \delta(A,S\cup\{(a,v)\},A,S)/c_{(a,v)} \right\}$\;\label{alg:greedy:r2}
  $r_3 = \max_{a\in U, (a,v)\in (\{a\}\times V)\cap T} \left\{ \delta(A\cup \{a\},S\cup\{(a,v)\},A,S)/(1+c_{(a,v)}) \right\}$\;\label{alg:greedy:r3}
  \eIf{$\max\{r_1,r_2,r_3\} = r_1$}
  {
    $\hat{a} = \arg\max_{a\in U} \{ \delta(A\cup \{a\},S,A,S) \}$\;
    \If{$k-1\geq 0$}
    {
      $A := A \cup \{ \hat{a} \}$\;
      $k := k-1$\;
    }
    $U := U\setminus \{ \hat{a} \}$\;
  }
  {
    \eIf{$\max\{r_1,r_2,r_3\} = r_2$}
    {
      $(\hat{a},\hat{v}) : = \arg\max_{(a,v)\in (A\times V)\cap T} \left\{ \delta(A,S\cup\{(a,v)\},A,S)/c_{(a,v)} \right\}$\;
      \If{$k-c_{(\hat{a},\hat{v})}\geq 0$}
      {
        $S := S \cup \{ (\hat{a},\hat{v}) \}$\;
        $k := k-c_{(\hat{a},\hat{v})}$\;
      }
      $T := T \setminus \{ (\hat{a},\hat{v}) \}$\;
    }
    {
      $(\hat{a},(\hat{a},\hat{v})) := \arg\max_{a\in U, (a,v)\in (\{a\}\times V)\cap T} \left\{ \delta(A\cup \{a\},S\cup\{(a,v)\},A,S)/(1+c_{(a,v)}) \right\}$\;
      \If{$k-c_{(\hat{a},\hat{v})}-1\geq 0$}
      {
        $(A,S) := (A \cup \{ \hat{a}\},S \cup \{ (\hat{a},\hat{v}) \}) $\;
%         $S := S \cup \{ (\hat{a},\hat{v}) \}$\;
        $U := U\setminus \{ \hat{a} \}$\;
        $k := k-1-c_{(\hat{a},\hat{v})}$\;
      }
      $T := T \setminus \{ (\hat{a},\hat{v}) \}$\;
    }
  }
}\label{alg:greedy:whileend}
$(a_{M},(a_{M},v_{M})) : = \arg\max_{a\in V, (a,v)\in (\{a\}\times V)\setminus E} \left\{ \sigma(A\cup \{a\},S\cup\{(a,v)\}) \right\}$\;\label{alg:greedy:max}
\Return $\arg\max\{ \sigma(A,S), \sigma(\{a_M\}, \{(a_{M},v_{M})\}) \}$\;\label{alg:greedy:return}
\end{algorithm2e}
Our algorithm, whose pseudocode is reported in Algorithm~\ref{alg:greedy}, finds two candidate solutions: the first solution is obtained by a greedy algorithm at lines~\ref{alg:greedy:whilestart}--\ref{alg:greedy:whileend}, the second solution is found at line~\ref{alg:greedy:max} and is made of a single node $a_{M}$ and a single edge $(a_{M},v_{M})$ for which $\sigma(\{a_M\},\{(a_M,v_M)\})$ is maximized. Then, the algorithm outputs one of the candidate solutions that maximizes the expected number of affected nodes (line~\ref{alg:greedy:return}).

The greedy phase, at each iteration, selects a solution $(A,S)$ that adds at most one node and one edge to the current solution $(A',S')$ and that maximizes the ratio between $\delta(A,S,A',S')$ and the marginal cost of $(A,S)$, that is the cost of the added node or edge. In particular, it considers three possible ways of obtaining $(A,S)$ from $(A',S')$:
\begin{description}
 \item[line~\ref{alg:greedy:r1}:] select a seed node ${a}$ that maximizes $r_1=\delta(A'\cup \{{a}\},S',A',S')$, $(A,S) = (A'\cup \{{a}\},S')$; %in this case the cost of adding node ${a}$ to $A$ is equal to 1, $c(A\setminus A',S\setminus S') = 1$;
 \item[line~\ref{alg:greedy:r2}:] select an edge $({a},{v})$ incident to a seed $a$ in $A'$ that maximizes \\$r_2=\frac{\delta(A',S'\cup\{({a},{v})\},A',S')}{c_{({a},{v})}}$, $(A,S) = (A',S'\cup\{({a},{v})\})$; %in this case $c(A\setminus A',S\setminus S') = c_{({a},{v})}$;
 \item[line~\ref{alg:greedy:r3}:]  select a seed node ${a}$ not in $A'$ and edge $({a},{v})$ incident to ${a}$ that maximize $r_3 = \frac{\delta(A'\cup \{{a}\},S'\cup\{({a},{v})\},A',S')}{1+c_{({a},{v})}}$, $(A,S) = (A'\cup \{{a}\},S'\cup\{({a},{v})\})$.% in this case the cost $c(A\setminus A',S\setminus S') = 1 + c_{(a,v)}$.
\end{description}
The greedy phase of the algorithm selects a solution $(A,S)$ that maximizes the three above ratios. If $(A,S)$ does not violate the budget, i.e. cost $c(A,S)$ is at most $k$, then it is chosen as new solution.

We denote by $(A^*,S^*)$ an optimal solution to the \MICM problem.  Let us consider the iterations executed by the greedy algorithm in which
an element is added to $(A,S)$. For $i\geq 1$, let us denote by $j_i$ the index of these iterations, $j_i<j_{i+1}$, and let $j_{l+1}$ be the index of the first iteration in which an element in $(A^*,S^*)$ is considered (i.e. it maximizes the above ratios) but not added to $(A,S)$ because it violates the budget constraint.
% the element added to $(A,S)$ belongs to $(A^*,S^*)$. For $i\geq 1$, let us denote by $j_i$ the index of such iterations, $j_i<j_{i+1}$, and let $j_{l+1}$ be the index of the first iteration in which an element in $(A^*,S^*)$ is considered but not added to $(A,S)$. That is in all the $l$ iterations $j_1,j_2,\ldots,j_l$ the greedy algorithm adds to $(A,S)$ an element in $(A^*,S^*)$.
We denote by $(A_i,S_i)$ the solution at the end of iteration $j_i$ and by $\bar{c}_i$ the marginal cost of $(A_i,S_i)$ as computed in the above three ratios, 
\[
 \bar{c}_i =\left\{
 \begin{array}{llll}
   1   && \mbox{ if } A_i\setminus A_{i-1}=\{{a}\} \mbox{ and } S_i = S_{i-1} & \mbox{(i.e. }\max\{r_1,r_2,r_3\} = r_1\mbox{)}\\
     c_e &&\mbox{ if } A_i=A_{i-1} \mbox{ and } S_i\setminus S_{i-1}=\{({a},{v})\} &\mbox{(i.e. }\max\{r_1,r_2,r_3\} = r_2\mbox{)}\\ 
    1+ c_e &&\mbox{ if } A\setminus A_{i-1}=\{{a}\}  \mbox{ and } S_i\setminus S_{i-1}=\{({a},{v})\}&\mbox{(i.e. }\max\{r_1,r_2,r_3\} = r_3\mbox{)}. 
 \end{array}
 \right.
\]
The next two lemmas are the core of our analysis~\cite{KMN99}.
%, note that the statement is similar to Lemma~1 in~\cite{KMN99} but....
\begin{lemma}\label{lem:marginalcosts_min}
 After each iteration $j_i$, $i=1,2,\ldots,l+1$, $$\sigma(A_i,S_i) - \sigma(A_{i-1},S_{i-1}) \geq \frac{\bar{c}_i}{k}\frac{c_{\min}}{1+c_{\min}} (\sigma(A^*,S^*) - \sigma(A_{i-1},S_{i-1})).$$
\end{lemma}
\begin{proof}
We denote by $\delta_i$ the expected number of nodes affected by solution $(A_i,S_i)$ and not affected by solution $(A_{i-1},S_{i-1})$, $\delta_i=\delta(A_i,S_i,A_{i-1},S_{i-1})$.
We first show that  the value $\sigma(A^*,S^*) - \sigma(A_{i-1},S_{i-1})$ is at most the sum, for each element in $(A^*,S^*)$ and not in $(A_{i-1},S_{i-1})$, of the expected number of nodes affected by this element and not affected by solution $(A_{i-1},S_{i-1})$, that is the following inequality holds:
% \begin{equation}\label{eq:one_min}
% \begin{split}
%  & \sigma(A^*,S^*) - \sigma(A_{i-1},S_{i-1}) \leq \sum_{a\in A^*_1} \delta(A_{i-1}\cup\{a\}, S_{i-1} , A_{i-1},S_{i-1}) +  \\ 
% % 
%  &\sum_{\begin{subarray}{c}e=(a,v)\in S^*\setminus S_{i-1}\\ \text{s.t. } a\in A_{i-1} \end{subarray} }\!\!\!\!\!\!\!\!\!\!\!\! \delta(A_{i-1}, S_{i-1}\cup\{e\}, A_{i-1},S_{i-1})
% +\!\!\!\!\!\!\!\! \sum_{\begin{subarray}{c}a\in A^*_2\\ e=(a,v)\in S^*\setminus S_{i-1} \end{subarray} }\!\!\!\!\!\!\!\!\!\!\!\! \delta(A_{i-1}\cup\{a\}, S_{i-1}\cup\{e\}, A_{i-1},S_{i-1}),
% \end{split}
% \end{equation}
\begin{align}
 & \!\!\!\!\sigma(A^*,S^*) - \sigma(A_{i-1},S_{i-1}) \leq \sum_{a\in A^*_1} \delta(A_{i-1}\cup\{a\}, S_{i-1} , A_{i-1},S_{i-1}) + \label{eq:one_min} \\ 
 & \!\!\!\!\sum_{\begin{subarray}{c}e=(a,v)\in S^*\setminus S_{i-1}\\ \text{s.t. } a\in A_{i-1} \end{subarray} }\!\!\!\!\!\!\!\!\!\!\!\! \delta(A_{i-1}, S_{i-1}\cup\{e\}, A_{i-1},S_{i-1})
+\!\!\!\!\!\!\!\! \sum_{\begin{subarray}{c}a\in A^*_2\\ e=(a,v)\in S^*\setminus S_{i-1} \end{subarray} }\!\!\!\!\!\!\!\!\!\!\!\! \delta(A_{i-1}\cup\{a\}, S_{i-1}\cup\{e\}, A_{i-1},S_{i-1}),\nonumber
\end{align}
where we divided the set $A^*\setminus A_{i-1}$ into two subsets: $A^*_1$ is the subset of $A^*\setminus A_{i-1}$ that contains the seeds $a$ with no incident edges in $S^*$ (i.e. $\nexists (a,v)\in S^*$), and $A^*_2= A^*\setminus (A_{i-1}\cup A^*_1)$.
% Recall that the cost of each edge is bigger than the quantity $c_{\min}$.
% 
% 
% Since $S_{i-1}\subseteq S^*$, then for each live-edge graph $X\in\chi(S_{i-1})$ there exist $2^{|S^*|-|S_{i-1}|}$ live-edge graph $Y\in \chi(S^*,X)$ such that $X$ is a subgraph of $Y$. Therefore,
The difference $\sigma(A^*,S^*) - \sigma(A_{i-1},S_{i-1})$ is at most $\sigma(A^*\cup A_{i-1},S^* \cup S_{i-1}) - \sigma(A_{i-1},S_{i-1})$ and therefore upper-bounded by:
\begin{align}
     & \sum_{X \in \chi(S^*\cup S_{i-1})} \mathbb{P}[X]|R(A^*\cup A_{i-1}, X)| - \sum_{X \in \chi(S_{i-1})} \mathbb{P}[X]|R(A_{i-1}, X)|\nonumber\\
     &= \sum_{X \in \chi(S_{i-1})} \sum_{Y\in\chi(S^*\cup S_{i-1},X)}\mathbb{P}[Y]|R(A^*\cup A_{i-1}, Y)| - \sum_{X \in \chi(S_{i-1})} \mathbb{P}[X]|R(A_{i-1}, X)|\nonumber\\
     &= \sum_{X \in \chi(S_{i-1})}\mathbb{P}[X] \sum_{Y\in\chi(S^*\cup S_{i-1},X)}\mathbb{P}[Y|X]|R(A^*\cup A_{i-1}, Y)| - \sum_{X \in \chi(S_{i-1})} \mathbb{P}[X]|R(A_{i-1}, X)|\nonumber\\
     &= \sum_{X \in \chi(S_{i-1})}\mathbb{P}[X]\left[ \sum_{Y\in\chi(S^*\cup S_{i-1},X)}\mathbb{P}[Y|X]|R(A^*\cup A_{i-1}, Y)| - |R(A_{i-1}, X)|  \right]\nonumber\\
     &= \sum_{X \in \chi(S_{i-1})}\!\!\mathbb{P}[X]\left[ \sum_{Y\in\chi(S^*\cup S_{i-1},X)}\!\!\!\!\!\!\!\!\!\!\!\mathbb{P}[Y|X]|R(A^*\cup A_{i-1}, Y)| - \!\!\!\!\sum_{Y\in\chi(S^*\cup S_{i-1},X)}\!\!\!\!\!\!\!\!\!\!\!\!\mathbb{P}[Y|X]|R(A_{i-1}, X)|  \right]\nonumber\\
     &= \sum_{X \in \chi(S_{i-1})}\mathbb{P}[X]\sum_{Y\in\chi(S^*\cup S_{i-1},X)}\mathbb{P}[Y|X]\left( |R(A^*\cup A_{i-1}, Y)| - |R(A_{i-1}, X)|\right)\label{eq:two_min}
\end{align}
For each $X \in \chi(S_{i-1})$ and $Y\in\chi(S^*\cup S_{i-1},X)$, \\the difference $|R(A^*\cup A_{i-1}, Y)| - |R(A_{i-1}, X)|$ between the nodes reachable from $A^*\cup A_{i-1}$ in $Y$ and those reachable from $A_{i-1}$ in $X$ can be bounded as follows:
% is at most the sum of:
% \begin{enumerate}
% \item the difference $|R(A_{i-1}\cup\{a\}, X)| - |R(A_{i-1}, X)|$, for each seed $a$ in $A^*$ and not in $A_{i-1}$,
% \item the difference $|R(A_{i-1}, X\cup\{e\})| - |R(A_{i-1}, X)|$, for each edge $e$ in $Y$ and not in $X$,
% \item the difference $|R(A_{i-1} \cup\{a\}, X\cup\{e\})| - |R(A_{i-1}\cup\{a\}, X)|$, for each edge $e$ in $Y$ and not in $X$,
% \end{enumerate}
% that is:
% \begin{equation}\label{eq:three_min}
% \begin{split}
% &|R(A^*\cup A_{i-1}, Y)| - |R(A_{i-1}, X)| \leq \sum_{ a \in A^*_1}\left( |R(A_{i-1}\cup\{a\}, X)| - |R(A_{i-1}, X)|\right) \\ 
% &+\!\!\!\!\!\sum_{\begin{subarray}{c} e=(a,v)\in Y\setminus X \\ \text{s.t. } a\in A_{i-1} \end{subarray}}\!\!\!\!\!\!\!\!\!\left(|R(A_{i-1}, X\cup\{e\})| - |R(A_{i-1}, X)| \right)
% +\!\!\!\!\!\sum_{\begin{subarray}{c} a\in A^*_2\\ e=(a,v)\in Y\setminus X \\ \end{subarray}}\!\!\!\!\!\!\!\!\!\!\! \left(|R(A_{i-1}\cup\{a\}, X\cup\{e\})| - |R(A_{i-1}, X)| \right).
% \end{split}
% \end{equation}
\begin{align}
\!\!\!\!\!\!\!\!\!|R(A^*\cup A_{i-1}, Y)| - |R(A_{i-1}, X)| &\leq \sum_{ a \in A^*_1}\left( |R(A_{i-1}\cup\{a\}, X)| - |R(A_{i-1}, X)|\right) \label{eq:three_min}\\ 
&\!\!\!\!\!\!\!\!\!+\sum_{\begin{subarray}{c} e=(a,v)\in Y\setminus X \\ \text{s.t. } a\in A_{i-1} \end{subarray}}\!\!\!\!\!\!\!\!\!\!\!\left(|R(A_{i-1}, X\cup\{e\})| - |R(A_{i-1}, X)| \right) \nonumber \\
&\!\!\!\!+\sum_{\begin{subarray}{c} a\in A^*_2\\ e=(a,v)\in Y\setminus X \\ \end{subarray}}\!\!\!\!\!\!\!\!\!\!\! \left(|R(A_{i-1}\cup\{a\}, X\cup\{e\})| - |R(A_{i-1}, X)| \right).\nonumber
\end{align}
Combining~\eqref{eq:two_min} and the first term of \eqref{eq:three_min}, we obtain the first term of~\eqref{eq:one_min}. To show the second and third term of~\eqref{eq:one_min}, observe that for a function $f:V\times V\rightarrow \mathbb{N}$ and for each $X \in \chi(S_{i-1})$,
\[
 \sum_{Y\in\chi(S^*\cup S_{i-1},X)}\!\!\!\!\!\!\!\! \mathbb{P}[Y|X] \sum_{e\in Y\setminus X} f(e) \leq \!\!\!\!\sum_{e\in S^*\setminus S_{i-1}}\!\!p_e \sum_{Y\in\chi(S^*\cup S_{i-1}\setminus\{e\},X)} \!\!\!\!\!\!\!\!\!\! \mathbb{P} [Y|X\cup\{e\}]f(e) = \!\! \sum_{e\in S^*\setminus S_{i-1}}\!\!\!\!p_e f(e).
\]
This shows inequality~\eqref{eq:one_min}.

Since the greedy phase of the algorithm selects a solution that that maximizes the ratio between the marginal increment in the objective function and the cost, the following holds:
\begin{itemize}
 \item For each $a\in A^*_1$, $\delta(A_{i-1}\cup\{a\}, S_{i-1} , A_{i-1},S_{i-1}) \leq \frac{\delta_i}{\bar{c}_i}$;
 \item For each $e=(a,v)\in S^*\setminus S_{i-1} $ such that \\$a\in A_{i-1}$, $\frac{\delta(A_{i-1}, S_{i-1}\cup\{e\}, A_{i-1},S_{i-1})}{c_e} \leq \frac{\delta_i}{\bar{c}_i}$;
 \item For each $a\in A^*_2$ and $e=(a,v)\in S^*\setminus S_{i-1}$,\\ $\frac{\delta(A_{i-1}\cup\{a\} S_{i-1}\cup\{e\}, A_{i-1},S_{i-1})}{1+c_e} \leq \frac{\delta_i}{\bar{c}_i}$.
\end{itemize}
Since the edge costs are at least $c_{\min}$, then the number of edges in $S^*\setminus S_{i-1}$ incident to each $a \in A^*_2$ is at most $\big\lfloor\frac{k}{c_{\min}}\big \rfloor \leq \frac{k}{c_{\min}}$. % and each of them has cost at most $1+c_e$.
Therefore, the right hand side of~\eqref{eq:one_min} is at most:
\begin{align*}
 &\sum_{a\in A^*_1} \frac{\delta_i}{\bar{c}_i} 
 + \sum_{\begin{subarray}{c}e=(a,v)\in S^*\setminus S_{i-1}\\ \text{s.t. } a\in A_{i-1} \end{subarray}} \frac{\delta_i}{\bar{c}_i}c_e
  + \sum_{\begin{subarray}{c}\\a\in A^*_2\\ e=(a,v) \in S^*\setminus S_{i-1}\end{subarray}} \frac{\delta_i}{\bar{c}_i}(1+c_e) 
\\ &\leq  \frac{\delta_i}{\bar{c}_i} \left(|A^*_1| 
  + \sum_{\begin{subarray}{c}e=(a,v)\in S^*\setminus S_{i-1}\\ \text{s.t. } a\in A_{i-1} \end{subarray}} c_e 
+\frac{k}{c_{\min}}+ \sum_{\begin{subarray}{c}a\in A^*_2\\ e=(a,v)\in S^*\setminus S_{i-1}\end{subarray}}c_e\right)
\leq \left(1+\frac{1}{c_{\min}}\right)k \frac{\delta_i}{\bar{c}_i}.
\end{align*}
Where the last inequality is due to  $$\displaystyle|A^*_1| 
  + \sum_{\begin{subarray}{c}e=(a,v)\in S^*\setminus S_{i-1}\\ \text{s.t. } a\in A_{i-1} \end{subarray}} c_e + \sum_{\begin{subarray}{c}a\in A^*_2\\ e=(a,v)\in S^*\setminus S_{i-1}\end{subarray}}c_e \leq k.$$

To conclude the proof, $\delta_i =  \sigma(A_i,S_i) - \sigma(A_{i-1},S_{i-1})$ follows from Proposition~\ref{prop:delta}.
% Indeed, if $S_i\setminus S_{i-1}=\{e\}$,
% \begin{align*}
%  \delta_i &= \sum_{X \in \chi(S_i)} \mathbb{P}[X]\cdot\left(|R(A_i, X)| - |R(A_{i-1}, X^e)|  \right)\\
%           &= \sigma(A_i,S_i) - \sum_{X \in \chi(S_{i-1})} \left(p_e|R(A_{i-1}, X| + (1-p_e)|R(A_{i-1}, X|\right)\\
%           &= \sigma(A_i,S_i) - \sigma(A_{i-1},S_{i-1}).\qedhere
% \end{align*}
\end{proof}

\begin{lemma}\label{lem:induction_min}
 After each iteration $j_i$, $i=1,2,\ldots,l+1$, $$\sigma(A_i,S_i) \geq \left[ 1- \prod_{\ell=1}^i\left( 1 - \frac{\bar{c}_\ell}{k}\frac{c_{\min}}{(1+c_{\min})}\right) \right]\sigma(A^*,S^*).$$
\end{lemma}
\begin{proof}
 We show the statement by induction on iterations $j_i$.
 For $i=1$, by Lemma~\ref{lem:marginalcosts_min}, $\sigma(A_1,S_1)  \geq \frac{\bar{c}_1}{k}\frac{c_{\min}}{(1+c_{\min})} \sigma(A^*,S^*) =\left[ 1- \left( 1 - \frac{\bar{c}_1}{k}\frac{c_{\min}}{(1+c_{\min})}\right) \right]\sigma(A^*,S^*)$. Let us assume that the statement holds for $j_1,j_2,\ldots,j_{i-1}$, then 
 \begin{align*}
 \!\!\!\!\!\!\!\!\! \sigma(A_i,S_i) &= \sigma(A_{i-1},S_{i-1}) + \left[ \sigma(A_i,S_i) - \sigma(A_{i-1},S_{i-1})\right]\\
                  &\geq \sigma(A_{i-1},S_{i-1}) + \frac{\bar{c}_i}{k}\frac{c_{\min}}{(1+c_{\min})}\left[ \sigma(A^*,S^*) - \sigma(A_{i-1},S_{i-1})\right]\\
                  &= \sigma(A_{i-1},S_{i-1}) \left(1-\frac{\bar{c}_i}{k}\frac{c_{\min}}{(1+c_{\min})}\right) + \frac{\bar{c}_i}{k}\frac{c_{\min}}{(1+c_{\min})}\sigma(A^*,S^*)\\
                  &\geq \left[ 1- \prod_{\ell=1}^{i-1}\left( 1 - \frac{\bar{c}_\ell}{k}\frac{c_{\min}}{(1+c_{\min})}\right) \right]\left(1-\frac{\bar{c}_i}{k}\frac{c_{\min}}{(1+c_{\min})}\right)\sigma(A^*,S^*)\\ &\quad\quad\quad\quad\quad\quad\quad\quad\quad\quad\quad\quad\quad\quad\quad\quad\quad\quad+\frac{\bar{c}_i}{k}\frac{c_{\min}}{(1+c_{\min})}\sigma(A^*,S^*)\\
                  &=\left[ 1- \prod_{\ell=1}^i\left( 1 - \frac{\bar{c}_\ell}{k}\frac{c_{\min}}{(1+c_{\min})}\right) \right]\sigma(A^*,S^*),
 \end{align*}
where the two inequalities follow from Lemma~\ref{lem:marginalcosts_min} and the inductive hypothesis, resp.
\end{proof}

\begin{theorem}\label{lem:final_mincost}
 Algorithm~\ref{alg:greedy} achieves an approximation factor of $$\frac{1}{2}\left(1-\frac{1}{e^{\frac{c_{\min}}{1+c_{\min}}}}\right)\sigma(A^*,S^*).$$
\end{theorem}
\begin{proof}
We observe that since $(A_{l+1},S_{l+1})$ violates the budget, then\\ $c(A_{l+1},S_{l+1})>k$. Moreover, for a sequence of numbers $a_1,a_2,\ldots,a_n$ such that $\sum_{\ell=1}^n a_\ell = B$, the function $\left[ 1 - \prod_{i=1}^n \left(1 -\frac{a_i}{B\cdot \beta}\right)\right]$ achieves its minimum when $a_i=\frac{B}{n}$ and that 
$
\left[ 1 - \prod_{i=1}^n \left(1 -\frac{a_i}{B\cdot \beta}\right)\right]\geq 1-\left(1-\frac{1}{n\cdot \beta}\right)^n\geq 1-e^{-\frac{1}{\beta}}.
$
Therefore, by applying Lemma~\ref{lem:induction_min} for $i=l+1$ and observing that $\sum_{\ell=1}^{l+1}\bar{c}_\ell=c(A_{l+1},S_{l+1})$, we obtain:
\begin{align*}
 \sigma(A_{l+1},S_{l+1}) & \geq \left[ 1- \prod_{\ell=1}^{l+1}\left( 1 - \frac{\bar{c}_\ell}{k\left(\frac{1+c_{\min}}{c_{\min}}\right)}\right) \right]\sigma(A^*,S^*)\\
                         & \geq \left[ 1- \prod_{\ell=1}^{l+1}\left( 1 - \frac{\bar{c}_\ell}{c(A_{l+1},S_{l+1})\left(\frac{1+c_{\min}}{c_{\min}}\right)}\right) \right]\sigma(A^*,S^*)\\
                         & \geq \left[ 1- \left( 1 - \frac{1}{(l+1)\frac{1+c_{\min}}{c_{\min}}}\right)^{l+1} \right]\sigma(A^*,S^*)\\
                          &\geq \left(1-\frac{1}{e^{\frac{c_{\min}}{1+c_{\min}}}}\right) \sigma(A^*,S^*).
\end{align*}
By Proposition~\ref{prop:delta}, it follows that: 
\begin{align}
\sigma(A_{l+1},S_{l+1}) = \sigma(A_{l},S_{l}) + \delta_{l+1}\geq \left(1-\frac{1}{e^{\frac{c_{\min}}{1+c_{\min}}}}\right) \sigma(A^*,S^*).\label{eq:approx}
\end{align}
 Since $\delta_{l+1} \leq \sigma(\{a_M\}, \{(a_{M},v_{M})\})$, we get
$$\sigma(A_{l},S_{l}) + \sigma(\{a_M\}, \{(a_{M},v_{M})\})\geq \left(1-\frac{1}{e^{\frac{c_{\min}}{1+c_{\min}}}}\right) \sigma(A^*,S^*).$$ Hence, $\max\{ \sigma(A_l,S_l), \sigma(\{a_M\}, \{(a_{M},v_{M})\}) \}\geq \frac{1}{2}\left(1-\frac{1}{e^{\frac{c_{\min}}{1+c_{\min}}}}\right) \sigma(A^*,S^*)$.
\end{proof}

\begin{algorithm2e}[t]
\caption{}
\label{alg:greedy_2}
\SetKwInput{Proc}{Algorithm}
%\Proc{\Greedy}
\SetKwInOut{Input}{Input}
\SetKwInOut{Output}{Output}
\Input{A directed graph $G=(V,E)$, integer $M\in\mathbb{N}$ and an integer $k\in\mathbb{N}$}
\Output{A set of nodes $A$ and a of edges $S\subseteq (A\times V)\setminus E$ such that $c(A,S)\leq k$}
$(A_1, S_1) : = \arg\max\{ \sigma(A,S): |A|+|S|< M, c(A, S)\leq k \}$\;\label{alg:greedy_2:first}
$A_2:=\emptyset$; $S_2:=\emptyset$;
$U:=V$;
$T:=(V\times V)\setminus E$\;
\ForEach{$A\subseteq U, S\subseteq (A\times V)\setminus E \text{\rm ~such that }|A|+|S|= M \text{\rm ~and } c(A, S)\leq k $}
{ 
$U := U\setminus A$;
$T := T \setminus S$\;
Complete $(A,S)$ by using Algorithm~\ref{alg:greedy} with $U$ and $T$ as possible nodes and edges\;
% $(A, S):=Greedy(G,k)$\tcp*{Call to \Greedy $c_{\min}$ algorithm}
\If{$\sigma(A, S)>\sigma(A_2, S_2)$}
{
	$A_2 := A$\;
	$S_2 := S$\;	
}
}
\Return $\arg\max\{ \sigma(A_1,S_1), \sigma(A_2, S_2)\}$\;
\end{algorithm2e}
We now propose an algorithm which improves the performance guarantee of Algorithm~\ref{alg:greedy}. 
Let $M$ a fixed integer, we consider all the solutions $(A,S)$ with cardinality $M$ (i.e. $|A|+|S|=M$) and cost at most $k$ ($c(A, S)\leq k$), and we complete all these solutions by using the greedy algorithm. The pseudocode is reported in Algorithm~\ref{alg:greedy_2}.

\begin{theorem}\label{th:enumeration}
 For $M\geq 4$ Algorithm~\ref{alg:greedy_2} achieves an approximation factor of  $1-\frac{1}{e^{\frac{c_{\min}}{1+c_{\min}}}}$.
\end{theorem}
\begin{proof}
We assume that $|A^*|+|S^*|>M$ since otherwise Algorithm~\ref{alg:greedy_2} finds an optimal solution.
We sort the elements in $(A^*, S^*)$ by selecting at each step the element, which can be either a seed or an edge, that maximizes the marginal increment in the number of influenced nodes.
Let $Z=(A_Z, S_Z)$ be the first $M$ elements in this order. We now consider the iteration of Algorithm~\ref{alg:greedy_2} in which element $Z$ is considered. We define $(A_{Z'},S_{Z'})$ as the elements added by the algorithm to $(A_Z,S_Z)$ and $(A, S)=(A_Z \cup A_{Z'}, S_Z \cup S_{Z'})$. By Proposition~\ref{prop:delta}, it follows that
$
\sigma(A, S)=\sigma(A_Z, S_Z)+\delta(A_Z \cup A_{Z'}, S_Z \cup S_{Z'}, A_Z, S_Z).
$

The completion of $(A_Z, S_Z)$ to $(A,S)$ is an application of the greedy algorithm and therefore, we can use the result from the previous theorems.  
Let us consider the iterations executed by the greedy algorithm during the completion of $(A_Z, S_Z)$ to $(A,S)$. For $i\geq 1$, let us denote by $j_i$ the index of these iterations, $j_i<j_{i+1}$, and let $j_{l+1}$ be the index of the first iteration in which an element in $(A^*\setminus A_Z,S^* \setminus S_Z)$ is considered but not added to $(A_{Z'}, S_{Z'})$ because it violates the budget constraint. Applying inequality~\eqref{eq:approx} to the instance of the problem obtained removing the nodes covered by solution $Z$ we get:
$
\delta(A_Z \cup A_{Z'}, S_Z \cup S_{Z'}, A_Z, S_Z) + \delta_{l+1}\geq \left(1-\frac{1}{e^{\frac{c_{\min}}{1+c_{\min}}}}\right) \sigma(A^*\setminus A_Z,S^* \setminus S_Z).
$
% notice that inequality~\eqref{eq:approx_2} follows by applying inequality~\eqref{eq:approx} to the instance of the problem obtained removing the nodes covered by solution $Z$.
% : removing such nodes we obtain  $\delta(A_Z \cup A_{Z'}, S_Z \cup S_{Z'}, A_Z, S_Z)$, in the first term of ~\eqref{eq:approx}, and $(A^*\setminus A_Z,S^* \setminus S_Z)$, in the second term. 
% 
Moreover, since we ordered the elements in $(A^*,S^*)$ and in iteration $j_{l+1}$ and most 2 elements are selected, then $\delta_{l+1}\leq \frac{2\sigma(A_Z, S_Z)}{M}$ and
\begin{align*}
\sigma(A, S)&=\sigma(A_Z, S_Z)+\delta(A_Z \cup A_{Z'}, S_Z \cup S_{Z'}, A_Z, S_Z)\\
 & \geq \sigma(A_Z, S_Z)+ \left(1-\frac{1}{e^{\frac{c_{\min}}{1+c_{\min}}}}\right) \sigma(A^*\setminus A_Z,S^* \setminus S_Z) -\delta_{l+1}\\
& \geq \sigma(A_Z, S_Z)+ \left(1-\frac{1}{e^{\frac{c_{\min}}{1+c_{\min}}}}\right) \sigma(A^*\setminus A_Z,S^* \setminus S_Z) -\frac{2\sigma(A_Z, S_Z)}{M}\\
&\geq \left(1-\frac{2}{M}\right) \sigma(A_Z, S_Z)+ \left(1-\frac{1}{e^{\frac{c_{\min}}{1+c_{\min}}}}\right) \sigma(A^*\setminus A_Z,S^* \setminus S_Z)
\end{align*}
But, $\sigma(A_Z, S_Z)+ \sigma(A^*\setminus A_Z,S^* \setminus S_Z)\geq \sigma(A^*, S^*)$, and we get:
\[
\sigma(A, S)\geq  \left(1-\frac{1}{e^{\frac{c_{\min}}{1+c_{\min}}}}\right) \sigma(A^*,S^*)+ \left(\frac{1}{e^{\frac{c_{\min}}{1+c_{\min}}}}-\frac{2}{M}\right) \sigma(A_Z,S_Z)\]
$$\geq  \left(1-\frac{1}{e^{\frac{c_{\min}}{1+c_{\min}}}}\right) \sigma(A^*,S^*),
$$
 for $M\geq 2e^{\frac{c_{\min}}{1+c_{\min}}}$. Since $2e^{\frac{c_{\min}}{1+c_{\min}}} < 4$ for $c_{\min}\in [0,1]$, the theorem follows.
\end{proof}

\section{General case}\label{sec:algo}

In this section we introduce an approximation algorithm for the \MICM problem that guarantees a constant approximation ratio, i.e. it does not depend on the edge costs.

The algorithm is similar to Algorithm~\ref{alg:greedy}, that is it chooses the best between two candidate solutions: one solution is made of a single node $a_M$ and a single edge $(a_M,v_M)$ that maximizes $\sigma(\{a_M\},\{(a_M,v_M)\})$, the other solution is found by means of a greedy algorithm that at each step maximizes the ratio between the marginal increment in the objective function of a candidate solution and its marginal cost. The main difference consists in the way in which the greedy algorithm computes a candidate solution. In fact, the algorithm presented here might select more than one edge at one time. In particular, at each iteration the greedy algorithm computes four candidate solutions starting from the current solution $(A',S')$. The first two solutions correspond to cases $r_1$ and $r_2$ of Algorithm~\ref{alg:greedy}, and differ from $(A',S')$ by a single seed node or a single edge incident to a seed in $A'$, respectively. To compute the two further candidate solutions, the algorithm divides the edges into two sets: those whose cost is at least $b$ and those whose cost is smaller than $b$, for some given constant $b$, and proceeds as follows.

\begin{itemize}
 \item Select a seed node ${a}$ not in $A'$ and an edge $({a},{v})$, which cost is \\$c_{({a}, {v})}\geq b$, incident to ${a}$ that maximize $\frac{\delta(A'\cup \{{a}\},S'\cup\{({a},{v})\},A',S')}{1+c_{({a},{v})}}$;
 \item Select a seed node ${a}$ not in $A'$ and a set ${S}$ of edges incident to ${a}$, whose overall cost is smaller than $b$ (i.e. $\sum_{({a},{v}) \in {S}}{c_{({a},{v})}}<b$), that maximize $\frac{\delta(A'\cup \{{a}\},S'\cup {S},A',S')}{1+\sum_{({a},{v}) \in {S}}{c_{({a}, {v})}}}$.
\end{itemize}
The greedy phase of the algorithm selects the candidate solution that gives the maximum ratio among the above four possibilities.

To compute the fourth candidate solution we need to solve an optimization problem, which we call R\MICM (where R stands for ratio). In the following, we assume that we can exploit an $\alpha$-approximation algorithm for this sub-problem. At the end of this section we will introduce a suitable algorithm for R\MICM.

The analysis of the algorithm is similar to that in the previous section. In particular, the next lemma is the core of the analysis and corresponds to Lemma~\ref{lem:marginalcosts_min}. We use the same notation used in the previous section. We denote by $j_i$ an iteration of the greedy algorithm in which an element is added to the solution, by $j_{l+1}$ the first iteration in which an element of the optimum is not added to the solution and by $\bar{c}_i$ the marginal cost of the solution $(A_i,S_i)$ computed at the iteration $j_i$.

\begin{lemma}\label{lem:marginalcosts}
 After each iteration $j_i$, $i=1,2,\ldots,l+1$, $$\sigma(A_i,S_i) - \sigma(A_{i-1},S_{i-1}) \geq \frac{\bar{c}_i}{k}\frac{b\alpha}{b\alpha+b+\alpha +2} (\sigma(A^*,S^*) - \sigma(A_{i-1},S_{i-1})).$$
\end{lemma}
\begin{proof}
We denote by $\delta_i$ the expected number of nodes affected by solution $(A_i,S_i)$ and not affected by solution $(A_{i-1},S_{i-1})$, $\delta_i=\delta(A_i,S_i,A_{i-1},S_{i-1})$.
Moreover, we divide the set $A^*\setminus A_{i-1}$ into two subsets: $A^*_1$ is the subset of $A^*\setminus A_{i-1}$ that contains the seeds $a$ with no incident edges in $S^*$, and $A^*_2= A^*\setminus (A_{i-1}\cup A^*_1)$. For each seed $a\in A^*_2$ let us consider the subset of $S^*$ containing edges incident to $a$ such that $c_e < b$. We partition this set into sets of edges whose overall cost is smaller than $b$ in such a way that the number of sets in the partition is minimized. We denote by $\mathcal{S}_a$ this partition and observe that the overall number of sets in $\mathcal{S}_a$, for all $a\in A^*_2$, is at most $2\frac{k}{b}$ as it is equivalent to the minimum number of bins of size $b$ needed to pack a set of items of overall size $k$.

By using similar arguments as in Lemma~\ref{lem:marginalcosts_min}, we can show that 
\begin{align}
% \begin{split}
\sigma(A^*,S^*) - \sigma(A_{i-1},S_{i-1}) &\leq 
 \sum_{ a\in A^*_1} \delta(A_{i-1}\cup\{a\}, S_{i-1} , A_{i-1},S_{i-1})\nonumber\\
 &+ \sum_{\begin{subarray}{c}e=(a,v)\in S^*\setminus S_{i-1}\\ \text{s.t. } a\in A_{i-1} \end{subarray} } \delta(A_{i-1}, S_{i-1}\cup\{e\}, A_{i-1},S_{i-1})\label{eq:one}\\
&\!\!\!\!\!\!+\sum_{\begin{subarray}{c}a\in A^*_2\\ e=(a,v)\in S^*\setminus S_{i-1}\\c_e \geq b \end{subarray} }\!\!\!\!\!\!\!\!\!\!\!\!\!\! \delta(A_{i-1}\cup\{a\}, S_{i-1}\cup\{e\}, A_{i-1},S_{i-1}) \nonumber \\
&+\! \sum_{\begin{subarray}{c}a\in A^*_2\\ S\in \mathcal{S}_a \end{subarray} } \delta(A_{i-1}\cup\{a\}, S_{i-1}\cup S,A_{i-1},S_{i-1}).\nonumber
% \end{split}
\end{align}
Indeed,
% \begin{equation*}
     $$\sigma(A^*,S^*) - \sigma(A_{i-1},S_{i-1}) \leq$$ $$\sum_{X \in \chi(S_{i-1})}\mathbb{P}[X]\sum_{Y\in\chi(S^*\cup S_{i-1},X)}\mathbb{P}[Y|X]( |R(A^*\cup A_{i-1}, Y)| - |R(A_{i-1}, X)|)$$
% \end{equation*}
and for each $X \in \chi(S_{i-1})$ and $Y\in\chi(S^*\cup S_{i-1},X)$,
\begin{align*}%\label{eq:three_min}
\!\!\!\!\!\!\!\!\!|R(A^*\cup A_{i-1}, Y)| - |R(A_{i-1}, X)| &\leq \sum_{a\in A^*_1}\left( |R(A_{i-1}\cup\{a\}, X)| - |R(A_{i-1}, X)|\right)\\ 
&\!\!\!\!\!\!\!\!\!+\sum_{\begin{subarray}{c}e=(a,v)\in Y\setminus X\\ \text{s.t. } a\in A_{i-1} \end{subarray}} \left(|R(A_{i-1}, X\cup\{e\})| - |R(A_{i-1}, X)| \right)\\
&\!\!\!\!\!\!\!\!\!+\sum_{\begin{subarray}{c}a\in A^*_2\\ e=(a,v)\in  Y\setminus X\\c_e \geq b\end{subarray}}\!\!\!\!\!\!\!\!\!\! \left(|R(A_{i-1}\cup\{a\}, X\cup\{e\})| - |R(A_{i-1}, X)| \right)\\
&+\sum_{\begin{subarray}{c}a\in A^*_2\\ S\in\mathcal{S}_a\end{subarray}} \!\!\left(|R(A_{i-1}\cup\{a\}, X\cup (S\cap Y))| - |R(A_{i-1}, X)| \right).
\end{align*}

The following holds since the greedy phase of the algorithm selects a solution that maximizes the ratio between the marginal increment in the objective function and the cost:
\begin{itemize}
 \item For each $a\in A^*_1$, $\delta(A_{i-1}\cup\{a\}, S_{i-1} , A_{i-1},S_{i-1}) \leq \frac{\delta_i}{\bar{c}_i}$;
 \item For each $e=(a,v)\in S^*\setminus S_{i-1} $ such that  $a\in A_{i-1}$,\\ $\frac{\delta(A_{i-1}, S_{i-1}\cup\{e\}, A_{i-1},S_{i-1})}{c_e} \leq \frac{\delta_i}{\bar{c}_i}$;
 \item For each $a\in A^*_2$ and $e=(a,v)\in S^*\setminus S_{i-1}$ s.t. $c_e\geq b$, \\$\frac{\delta(A_{i-1}\cup\{a\}, S_{i-1}\cup\{e\}, A_{i-1},S_{i-1})}{1+c_e} \leq \frac{\delta_i}{\bar{c}_i}$;
 \item For each $a\in A^*_2$ and $S\in\mathcal{S}_a$,\\ $\frac{\delta(A_{i-1}\cup\{a\}, S_{i-1}\cup S, A_{i-1},S_{i-1})}{1+\sum_{e\in S}c_e} \leq \frac{\delta_i}{\bar{c}_i}\cdot \frac{1}{\alpha}$,
\end{itemize}
where the term $\frac{1}{\alpha}$ in the last inequality is due to the use of an $\alpha$-approximation algorithm in the computation of the fourth candidate solution of the greedy phase.
The number of edges incident to each $a\in A^*_2$  with cost at least $b$ is at most $\frac{k}{b}$. Therefore, the right hand side of~\eqref{eq:one} is at most:
\begin{align*}
 &\sum_{ a\in A^*_1} \frac{\delta_i}{\bar{c}_i} 
 + \sum_{\begin{subarray}{c} e=(a,v)\in S^*\setminus S_{i-1}\\ \text{s.t. } a\in A_{i-1}\end{subarray}} \frac{\delta_i}{\bar{c}_i}c_e
  + \sum_{\begin{subarray}{c}a\in A^*_2\\ e=(a,v)\in S^*\setminus S_{i-1}\\c_e \geq b\end{subarray}} \frac{\delta_i}{\bar{c}_i}(1+c_e) 
  + \sum_{\begin{subarray}{c} a\in A^*_2\\ S\in\mathcal{S}_a\end{subarray}}\frac{1}{\alpha} \frac{\delta_i}{\bar{c}_i}\left(1+\sum_{e\in S}c_e\right) 
\\ &= \frac{\delta_i}{\bar{c}_i} \Bigg( |A^*_1| 
  + \sum_{\begin{subarray}{c} e=(a,v)\in S^*\setminus S_{i-1}\\ \text{s.t. } a\in A_{i-1}\end{subarray}} c_e 
+\frac{k}{b}+ \sum_{\begin{subarray}{c}a\in A^*_2\\ e=(a,v)\in S^*\setminus S_{i-1}\\c_e \geq b\end{subarray}}c_e
+\frac{1}{\alpha} \frac{2k}{b}+\frac{1}{\alpha} \sum_{\begin{subarray}{c} a\in A^*_2\\ S\in\mathcal{S}_a\end{subarray}}\sum_{e\in S}c_e\Bigg)\\ 
&\leq \left(1+\frac{1}{b}+\frac{2}{b\alpha}+\frac{1}{\alpha}\right)k \frac{\delta_i}{\bar{c}_i}=\left(\frac{b\alpha+\alpha +2+b}{b\alpha}\right)k \frac{\delta_i}{\bar{c}_i}.
\end{align*}
%Substituting $\alpha=\left(1-\frac{1}{e}\right)\frac{1}{b+1}=\frac{e-1}{eb+e}$ 
%$$\frac{b\alpha+b+\alpha +1}{b\alpha}=\frac{eb^2+3eb-b+2e-1}{eb-b} $$
To conclude, we observe that by Proposition~\ref{prop:delta}, it follows that $$\delta_i =  \sigma(A_i,S_i) - \sigma(A_{i-1},S_{i-1}).$$
\end{proof}

\begin{lemma}\label{lem:induction}
 After each iteration $j_i$, $i=1,2,\ldots,l+1$, $$\sigma(A_i,S_i) \geq \left[ 1- \prod_{\ell=1}^i\left( 1 - \frac{\bar{c}_\ell}{k}\frac{b\alpha}{b\alpha+b+\alpha +2}\right) \right]\sigma(A^*,S^*).$$
\end{lemma}

\begin{proof}
 We show the statement by induction on iterations $j_i$.
 For $i=1$, by Lemma~\ref{lem:marginalcosts}, $\sigma(A_1,S_1)  \geq \frac{\bar{c}_1}{k}\frac{b\alpha}{b\alpha+b+\alpha +2} \sigma(A^*,S^*) =\left[ 1- \left( 1 - \frac{\bar{c}_1}{k}\frac{b\alpha}{b\alpha+b+\alpha +2}\right) \right]\sigma(A^*,S^*)$. Let us assume that the statement holds for $j_1,j_2,\ldots,j_{i-1}$, then 
 \begin{align*}
  \sigma(A_i,S_i) &= \sigma(A_{i-1},S_{i-1}) + \left[ \sigma(A_i,S_i) - \sigma(A_{i-1},S_{i-1})\right]\\
                  &\geq \sigma(A_{i-1},S_{i-1}) + \frac{\bar{c}_i}{k}\frac{b\alpha}{b\alpha+b+\alpha +2}\left[ \sigma(A^*,S^*) - \sigma(A_{i-1},S_{i-1})\right]\\
                  &= \sigma(A_{i-1},S_{i-1}) \left(1-\frac{\bar{c}_i}{k}\frac{b\alpha}{b\alpha+b+\alpha +2}\right) + \frac{\bar{c}_i}{k}\frac{b\alpha}{b\alpha+b+\alpha +2}\sigma(A^*,S^*)\\
                  &\geq \left[ 1- \prod_{\ell=1}^{i-1}\left( 1 - \frac{\bar{c}_\ell}{k}\frac{b\alpha}{b\alpha+b+\alpha +2}\right) \right]\left(1-\frac{\bar{c}_i}{k}\frac{b\alpha}{b\alpha+b+\alpha +2}\right)\sigma(A^*,S^*)\\&\quad\quad\quad\quad\quad\quad\quad\quad\quad\quad\quad\quad\quad\quad\quad\quad\quad\quad+ \frac{\bar{c}_i}{k}\frac{b\alpha}{b\alpha+b+\alpha +2}\sigma(A^*,S^*)\\
                  &=\left[ 1- \prod_{\ell=1}^i\left( 1 - \frac{\bar{c}_\ell}{k}\frac{b\alpha}{b\alpha+b+\alpha +2}\right) \right]\sigma(A^*,S^*),
 \end{align*}
where the two inequalities follows from Lemma~\ref{lem:marginalcosts} and the inductive hypothesis, respectively.
\end{proof}
Therefore, equipped with Lemma~\ref{lem:induction}, we can prove the next theorem.% which corresponds to Theorem~\ref{lem:final_mincost}.
\begin{theorem}\label{th:final}
 The greedy algorithm achieves an approximation factor of  $$1-\frac{1}{e^{\frac{b\alpha}{b\alpha+b+\alpha +2}}}.$$
\end{theorem}

\begin{proof}
We recall that, for a sequence of numbers $a_1,a_2,\ldots,a_n$ such that $\sum_{\ell=1}^n a_\ell = \beta \cdot B$, the function $\left[ 1 - \prod_{i=1}^n \left(1 -\frac{a_i}{B}\right)\right]$ achieves its minimum when $a_i=\frac{\beta\cdot B}{n}$ and that 
$$
\left[ 1 - \prod_{i=1}^n \left(1 -\frac{a_i}{B}\right)\right]\geq 1-\left(1-\frac{\beta}{n}\right)^n\geq 1-e^{-\beta}.
$$
To prove the theorem we analyse the following cases. 
\begin{itemize}
\item First, note that if there is a choice of a node and an edge $(\{a\}, \{(a,v)\})$ influencing a fraction of nodes greater than $\frac{1}{2}\sigma(A^*, S^*)$, i.e. $$\sigma(\{a\}, \{(a,v)\})\geq \frac{1}{2}\sigma(A^*, S^*),$$ then this solution, or a solution influencing a grater number of nodes, will be examined by the algorithm as a
candidate solution, resulting in a final solution having
the value of at least $\frac{1}{2}\sigma(A^*, S^*)$. 
\item Now assume that no element influences a number of nodes greater than $\frac{1}{2}\sigma(A^*, S^*)$. 
We notice that $c(A_l, S_l) \geq k-1-b$, otherwise we obtain the contradiction $\bar{c}_{l+1}>1+b$ which is false since we have that $\bar{c}_{l+1}\leq 1+b$. In particular, the contradiction follows from:\\ $c(A_l, S_l)<k-1-b$ and $c(A_{l+1}, S_{l+1})>k$.
Then, by Lemma~\ref{lem:induction}, 	  for $i=l$ we obtain:
\begin{align*}
 \sigma(A_l,S_l) & \geq \left[ 1- \prod_{\ell=1}^{l}\left( 1 - \frac{\bar{c}_\ell}{k\left(\frac{b\alpha+b+\alpha +2}{b\alpha}\right)}\right) \right]\sigma(A^*,S^*)\\
  & \geq \left[ 1- \prod_{\ell=1}^{l}\left( 1 - \frac{\bar{c}_\ell}{(c(A_l, S_l) + 1+ b)\left(\frac{b\alpha+b+\alpha +2}{b\alpha}\right)}\right) \right]\sigma(A^*,S^*)\\
                         & \geq \left[ 1- \left( 1 - \frac{1}{(l+1+b)\frac{b\alpha+b+\alpha +2}{b\alpha}}\right)^{l} \right]\sigma(A^*,S^*)\\
                          &\geq \left(1-\frac{1}{e^{\frac{b\alpha}{b\alpha+b+\alpha +2}}}\right) \sigma(A^*,S^*).
\end{align*}
\end{itemize} 
Thus, in each of the cases, a value of the solution produced by the algorithm is at least $\left(1-\frac{1}{e^{\frac{b\alpha}{b\alpha+b+\alpha +2}}}\right)\sigma(A^*,S^*)$, and the theorem follows.
\end{proof}

It remains to give an $\alpha$-approximation algorithm for the R\MICM problem which consists in selecting a seed node $a$ not in the current solution $(A',S')$ and a set $S$ of edges incident to $a$, which overall cost is smaller than $b$ (i.e. $\sum_{(a,v) \in S}{c_{(a,v)}}<b$), that maximize $\frac{\delta(A'\cup \{a\},S'\cup S,A',S')}{1+\sum_{(a,v) \in S}{c_{(a, v)}}}$.
To this aim we define the  \CICM ~\cite{TCS17} problem as follows. Given a directed graph $G=(V,E, p, c)$, an integer $k\in\mathbb{N}$ and a seed set $A$, find a set of edges $S\subseteq A\times V\setminus E$ such that $c(S)\leq k$ and $\sigma(A, S)$ is maximized.

Let us consider the instances of the \CICM problem where $A=\{a\}$ and let $m^*$ be the maximum, over all $a\in V\setminus A'$, among the optima of these instances. 
It is easy to show that $m^*$ is at least $b+1$ times the optimum of R\CICM. It has been shown that the \CICM problem can be approximated within a factor of $1-\frac{1}{e}$ by using a greedy algorithm~\cite{TCS17}. Therefore we obtain an overall approximation of $\alpha=\left(1-\frac{1}{e}\right)\frac{1}{b+1}$.
\begin{corollary}
There exists an algorithm that achieves an approximation factor $> 0.0878$ for the \MICM problem.
\end{corollary}
\begin{proof}
It follows by applying Theorem~\ref{th:final} with $\alpha=\left(1-\frac{1}{e}\right)\frac{1}{b+1}$ and optimizing over $b$.
\end{proof}

\section{Conclusions}
\begin{figure}[t]
\centering
 \scalebox{.9}{\input{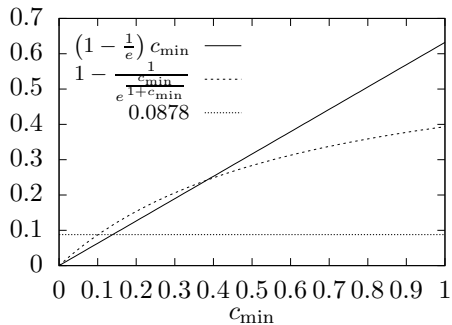}}
 \vspace{-4mm}
 \caption{Approximation ratio of the three algorithm presented as a function of $c_{\min}$. Intersection points are at $c_{\min}\approx0.1011$ and $c_{\min}\approx0.3821$}
 \label{fig:apxratio}
\end{figure}
In Figure~\ref{fig:apxratio} we summarize our approximation ratios as a function of $c_{\min}$. As expected, when all the edge costs are high, it is not worth to buy them. Indeed the algorithm that selects only seeds outperforms the other algorithms and it reaches the optimal approximation of $1-\frac{1}{e}$ when all the costs are equal to 1. However, the challenge of our problem consists in finding a good approximation even when the cost function includes small values. In the general case and when the minimum edge cost can be very small ($c_{\min}<0.1011$) the best algorithm is the constant factor algorithm given in Section~\ref{sec:algo}, while when the minimum cost is in $(0.1011,0.3821)$ the best algorithm is the one presented in Section~\ref{sec:algo_min}.

The main open problem is the gap between the lower bound on approximation of $1-\frac{1}{e}\approx 0.6321$ and the constant approximation of $0.0878$. To close this gap, we aim at improving the algorithm in Section~\ref{sec:algo} by devising a better approximation algorithm for R\CICM problem, which directly implies a better approximation for the \MICM problem.
Other research directions that deserve further investigation include the study of the \MICM problem on different information diffusion models such as LTM or the Triggering Model~\cite{KKT15}.

\newpage
\bibliographystyle{plainurl}% the recommended bibstyle
\bibliography{references}

\end{document}